\documentclass[graybox]{svmult}

\usepackage{mathptmx}       
\usepackage{helvet}         
\usepackage{courier}        
\usepackage{type1cm}        
\usepackage{makeidx}         
\usepackage{graphicx}        
\usepackage{multicol}        
\usepackage[bottom]{footmisc}

\usepackage{amssymb}
\usepackage{dsfont} 

\newcommand{\be}{\begin{equation}}
\newcommand{\ee}{\end{equation}}
\newcommand{\bel}[1]{\begin{equation}\label{#1}}
\newcommand{\bea}{\begin{eqnarray}}
\newcommand{\eea}{\end{eqnarray}}
\newcommand{\balign}{\begin{align}}
\newcommand{\ealign}{\end{align}}
\newcommand{\ba}{\begin{array}}
\newcommand{\ea}{\end{array}}
\newcommand{\bfig}{\begin{figure}}
\newcommand{\efig}{\end{figure}}
\newcommand{\eref}[1]{(\ref{#1})}

\newcommand{\bra}[1]{\mbox{$\langle \, {#1}\, |$}}
\newcommand{\ket}[1]{\mbox{$| \, {#1}\, \rangle$}}
\newcommand{\exval}[1]{\mbox{$\langle \, {#1}\, \rangle$}}
\newcommand{\inprod}[2]{\mbox{$\langle \, {#1} \, | \, {#2} \, \rangle$}}
\newcommand{\Prob}[1]{\mbox{${\rm Prob}\left[ \, {#1}\, \right]$}}
\newcommand{\comm}[2]{\mbox{$[\,{#1}\,,\,{#2}\,]$}}
\newcommand{\half}{\frac{1}{2}}

\newcommand{\bfx}{\vec{x}}
\newcommand{\bfy}{\vec{y}}
\newcommand{\bfz}{\vec{z}}

\newcommand{\rme}{\mathrm{e}}
\newcommand{\rmd}{\mathrm{d}}

\newcommand{\C}{{\mathbb C}}
\newcommand{\R}{{\mathbb R}}

\renewcommand{\S}{{\mathbb S}}

\makeindex             

\begin{document}

\title*{Quantum algebra symmetry of the ASEP with second-class particles}
\author{V. Belitsky and G.M.~Sch\"utz}
\institute{V. Belitsky \at Instituto de Matem\'atica e Est\'atistica,
Universidade de S\~ao Paulo, Rua do Mat\~ao, 1010, CEP 05508-090,
S\~ao Paulo - SP, Brazil, \email{belitsky@ime.usp.br}
\and G.M.~Sch\"utz \at Institute of Complex Systems II,
Forschungszentrum J\"ulich, 52425 J\"ulich, Germany \email{g.schuetz@fz-juelich.de}}
%
%
\maketitle

%
\abstract{We consider a two-component asymmetric simple exclusion process (ASEP) on a finite lattice
with reflecting boundary conditions. For this process, which is equivalent to the
ASEP with second-class particles, we construct the representation
matrices of the quantum algebra $U_q[\mathfrak{gl}(3)]$ that commute with the generator.
As a byproduct we prove reversibility and obtain in explicit form
the reversible measure. A review of the algebraic techniques used in the proofs is given.}

\section{Introduction}
\label{sec:1}

The standard asymmetric simple exclusion process (ASEP) \cite{Ligg85,Ligg99,Schu01} defined on a 
finite lattice  with reflecting boundary conditions
is a reversible process with explicitly known invariant measures. 
They are not translation invariant, in contrast to the non-reversible uniform measures for periodic
boundary conditions, and form the finite-size analogs of reversible blocking measures \cite{Ligg85}.
It has been shown in \cite{Sand94}
that these measures can be constructed by using a non-Abelian symmetry property of the generator,
viz. its commutativity with the generators of the quantum algebra $U_q[\mathfrak{gl}(2)]$. 

A related process of great interest is the ASEP with second class particles \cite{Ferr91}. For periodic
boundary conditions the invariant measures, which are non-reversible, can be computed in principle using 
the matrix product ansatz \cite{Derr93,Evan09,Prol09}, or
alternatively using methods from queuing theory \cite{Ferr07}. However, they have
a complicated non-uniform structure and no closed-form expression
for the stationary probabilities as a function of the particle configurations is known.

The invariant measures for the ASEP with second class particles with reflecting boundary 
have not been studied yet. However, it has been known for a long time that the generator of this
process has a quantum algebra symmetry, viz. its generator commutes with the generators of the 
quantum algebra $U_q[\mathfrak{gl}(3)]$ \cite{Alca93}. However, the corresponding representation 
matrices were never computed and the invariant measures for reflecting boundary conditions, which are
expected to be reversible measures, have remained unknown.
In this work we construct the matrix representations of the 
generators of $U_q[\mathfrak{gl}(3)]$
which commute with the generator of the ASEP with second-class particles.
This approach provides a constructive method to obtain in explicit form
reversible measures. We also review the algebraic
tools required in the proofs.

\section{Definitions and notation for the two-component ASEP}
\label{Sec:Definotat}

We set the stage and introduce some notation. 

\subsection{State space and configurations}
\label{Sec:Defconfig}

We consider the finite integer lattice $\Lambda := \{1,2,\dots,L\}$ of $L$ sites and local
occupation variables $\eta(k) \in \mathbb{S} = \{0,1,2\}$. We say that
a site $k \in \Lambda$ is occupied by a particle of type $A$ ($B$) if $\eta(k) = 0 (2)$ or that
it is empty (represented by the symbol $0$) if $\eta(k) = 1$. These local occupation variables
define the configuration $\eta =\{ \eta(1), \dots  ,\eta(L)\} \in \S^L$ of the particle system.
The fact that a site can be occupied by at most one particle of any type is the exclusion principle.

The following functions of configurations $\eta$ will play a role.
For $1\leq k \leq L$ we define the cyclic flip operation
\bel{flip}
\gamma^k(\eta) = \{ \eta(1), \dots \eta(k-1), \eta(k) + 1 (\mbox{mod } 3), \eta(k+1), \dots ,\eta(L) \}.
\ee
We define $\eta^{k\pm} := (\gamma^k)^{\pm 1}(\eta)$ and observe that $(\gamma^k)^{- 1} = (\gamma^k)^{2}$.
For $1\leq k \leq L-1$ We also define the local permutation
\bel{etaexch}
\sigma^{kk+1}(\eta) = \{ \eta(1), \dots \eta(k-1), \eta(k+1), \eta(k), \eta(k+2), \dots ,\eta(L) \}
=: \eta^{kk+1} .
\ee

We also define local occupation number variables
\bel{lonv}
a_k := \delta_{\eta(k),0}, \quad \upsilon_k := \delta_{\eta(k),1}, \quad b_k := \delta_{\eta(k),2}.
\ee
where $\delta_{k,l}$ is the usual Kronecker symbol with $k,l$ from any set.
In particular, we define the particle numbers
\bel{partnum}
N(\eta) = \sum_{k=1}^L a_k, \quad M(\eta) = \sum_{k=1}^L b_k, \quad V(\eta) = \sum_{k=1}^L \upsilon_k.
\ee
Notice that $N(\eta) + M(\eta) + M(\eta) = L$.
Occasionally we denote configurations with a fixed number $N$ particles of type $A$ and 
$M$ particles of type $B$ by $\eta_{N,M}$.\footnote{The occupation numbers can be formally regarded as 
families of mappings $a_k: \mathbb{S}^L \mapsto \{0,1\}$, $b_k: \mathbb{S}^L \mapsto \{0,1\}$ and 
should thus be understood as functions $a_k(\eta)$, $b_k(\eta)$ of $\eta$. 
Since the functional argument $\eta$ 
will always be clear from context (as is the case e.g. in \eref{partnum}), we do not write it 
explicitly. However, we shall usually write explicitly the argument for the particle
number functions $N(\eta)$, $M(\eta)$ to contrast them with their numerical values $N$, $M$.}

Another useful way to specify a configuration $\eta$ uniquely is
by indicating the particle 
positions on the lattice. We write
$\bfz(\eta) = \{ \bfx, \bfy \}$
with
$\bfx := \{ x \, :\, \eta(x) = 1 \}, \quad \bfy := \{ y \, :\, \eta(y) = 3 \}.$
We call this notation the position representation. Since the order of $A$-particles is conserved
we may label them 
consecutively from left to right by 1 to $N$, and similarly we may
label the $B$-particles by 1 to $M$.
By the exclusion principle one has
$\bfx \cap \bfy = \emptyset$ and by conservation of ordering
$ x_1 < x_2 < \dots < x_N$, $1 \leq y_1 < y_2 < \dots < y_M \leq L$. 
In multiple sums over $x_i$ and/or $y_i$ 
such sums will be understood as respecting all these exclusion constraints.
\footnote{We shall use interchangeably the arguments $\eta$, $\bfz$, 
$\{\bfx,\bfy\}$ for functions of the configurations. When the argument is clear from context
it may be omitted.}
We note the trivial, but frequently used identities
\bea
\label{partnum2}
N(\eta) \equiv N(\bfz) = |\bfx|, & & M(\eta) \equiv M(\bfz) = |\bfy| \\
\label{occupos}
a_k = \sum_{i=1}^{N(\bfz)} \delta_{x_i,k} , & &
 b_k  = \sum_{i=1}^{M(\bfz)} \delta_{y_i,k}.
\eea

For a configuration $\eta\equiv \bfz$ we also define 
the number $N_k(\eta)$ of 
$A$-particles to the left of a particle
at site $k$ and analogously 
the number $M_k(\eta)$ of $B$-particles and vacancies $V_k(\eta)$ to the left of site $k$
\bel{NyMx}
N_k(\eta) := \sum_{i=1}^{k-1} a_i = \sum_{i=1}^{N(\eta)} \sum_{l=1}^{k-1} \delta_{x_i,l}, \quad
M_k(\eta) := \sum_{i=1}^{k-1} b_i = \sum_{i=1}^{M(\eta)} \sum_{l=1}^{k-1} \delta_{y_i,l}.
\ee
Similarly we define $V_k(\eta) := \sum_{i=1}^{k-1} \upsilon_i$.

\subsection{The two-component ASEP}
\label{Sec:Defprocess}

Following \cite{Beli15} the two-component ASEP that we are going to study can be informally 
described as follows.
Each bond $(k,k+1)$, $1\leq k \leq L-1$ carries a clock which rings independently of all other clocks 
after an exponentially distributed
random time with parameter $\tau_k$ where $\tau_k=wq$ if $\eta(k+1)>\eta(k)$, 
$\tau_k=wq^{-1}$ if $\eta(k+1)<\eta(k)$ and $\tau_k=\infty$ if $\eta(k)=\eta(k+1)$. 
When the clock rings the 
particle occupation variables are interchanged and the clock acquires the 
parameter corresponding to interchanged variables.
Symbolically this process can be presented by the table of nearest neighbour particle
jumps
\be
\label{2ASEPr}
\left. \ba{l}
A0 \to 0A \\
0B \to B0 \\
AB \to BA \ea \right\}
\mbox{ with rate }  wq, \quad
\left. \ba{l}
0A \to A0 \\
B0 \to 0B \\
BA \to AB \ea \right\}
\mbox{ with rate }  wq^{-1} .
\ee
We consider reflecting boundary conditions, 
which means that no jumps from site 1 to the left and no jumps from site $L$ to the right are allowed.
We shall assume partially asymmetric hopping, i.e., $0 < q < \infty$. 
By interchanging the role of $B$-particles and vacancies this process turns into the
ASEP with second-class particles \cite{Ferr91}.

More precisely, for functions $f:\mathbb{S}^L \to \C$ we define this 
Markov process $\eta_t$ by the generator
\bel{generator}
\mathcal{L} f(\eta) := {\sum}_{\eta'\in\mathbb{S}^L}' w(\eta\to\eta') [f(\eta') - f(\eta)]
\ee
with the transition rates 
\bel{rates}
w(\eta\to\eta') = \sum_{k=1}^{L-1} w^{kk+1}(\eta) \delta_{\eta',\eta^{kk+1}}
\ee
defined in terms of the local hopping rates
\bel{localrates}
w^{kk+1}(\eta) = wq \left( a_k \upsilon_{k+1} + \upsilon_k b_{k+1} + a_k b_{k+1} \right) +
wq^{-1} \left(\upsilon_k a_{k+1} + b_k \upsilon_{k+1} + b_k a_{k+1} \right)
\ee
for a transition from a configuration 
$\eta$ to a configuration $\eta'=\eta^{kk+1}$ defined by \eref{etaexch}. 
The prime at the summation symbol \eref{generator} indicates the
absence of the term $\eta' = \eta$ which is omitted since $w(\eta\to\eta)$ is not
defined.\footnote{We shall usually omit the set 
$\mathbb{S}^L$ in the summation
symbol and simply write $\sum_{\eta}$.}

We fix more notation and summarize some well-known basic facts from the theory of Markov processes.
For a probability distribution $P(\eta)$ the expectation of a continuous function 
$f(\eta)$ is denoted by $\exval{f}_P := \sum_{\eta} f(\eta) P(\eta)$.
The transposed generator is defined by
$\mathcal{L}^T f(\eta) := {\sum}'_{\eta' \in\mathbb{S}^L} f(\eta') \mathcal{L} \, \mathbf{1}_{\eta'}(\eta)$
where $\mathbf{1}_{\eta'}(\eta) = \delta_{\eta,\eta'}$.
With this definition \eref{generator} yields for a probability distribution $P(\eta)$ the
{\it master equation} 
\bel{transgenerator}
\mathcal{L}^T P(\eta) = 
{\sum}'_{\eta'\in\mathbb{S}^L} [ w(\eta'\to\eta) P(\eta') - w(\eta\to\eta') P(\eta)].
\ee
The time-dependent probability distribution $P(\eta,t) := \Prob{\eta_t = \eta}$ follows
from the semi-group property
$P(\eta,t) = \rme^{\mathcal{L}^T t} P_0(\eta)$ with initial distribution $P_0(\eta):=P(\eta,0)$.
An invariant measure is denoted $\pi^\ast(\eta)$ and defined by 
\bel{invmeasure}
\mathcal{L}^T \pi^\ast(\eta) = 0, \quad \sum_{\eta} \pi^\ast(\eta) = 1.
\ee
A general stationary measure is denoted by $\pi$. It satisfies $\mathcal{L}^T \pi(\eta) = 0$,
but no assumption on the normalization $\sum_{\eta} \pi(\eta)$ is made.

The time-reversed process is defined by 
\bel{revgenerator}
\mathcal{L}^{rev} f(\eta) := {\sum_{\eta'}}' w^{rev}(\eta\to\eta') [f(\eta') - f(\eta)]
\ee
with 
$w^{rev}(\eta\to\eta') = w(\eta'\to\eta) \pi(\eta') / \pi(\eta)$.
The process is reversible if the rates satisfy the detailed balance condition
$w^{rev}(\eta\to\eta') = w(\eta'\to\eta)$. We remark that
\bel{revtrans}
\pi(\eta) \mathcal{L}^{rev}\, f(\eta) = \mathcal{L}^T (\pi(\eta) f(\eta))
\ee
which is a consequence of \eref{invmeasure}.

We define the transition matrix $H$ of the process by the matrix elements
\bel{transmatrix}
H_{\eta'\eta} = \left\{ \ba{ll} 
- w(\eta\to\eta') \quad & \eta \neq \eta' \\
{\sum}'_{\eta'} w(\eta\to\eta') & \eta = \eta' .
\ea \right.
\ee
with $w(\eta\to\eta')$ given \eref{rates}. In slight abuse of language we shall
also call $H$ the generator of the process.\\

\subsection{The quantum algebra $U_q[\mathfrak{gl}(n)]$}

The quantum algebra
$U_q[\mathfrak{gl}(n)]$ is the $q$-deformed universal enveloping algebra of 
the Lie algebra $\mathfrak{gl}(n)$. This associative algebra over $\C$ is
generated by $\mathbf{L}_i^{\pm 1}$, $i=1,\dots,n$ and $\mathbf{X}^\pm_i$, $i=1,\dots,n-1$
with the relations \cite{Jimb86,Burd92}
\bea
\label{Uqglndef1}
& & \comm{\mathbf{L}_i}{\mathbf{L}_j} = 0  \\
\label{Uqglncomm2}
& & \mathbf{L}_i \mathbf{X}^\pm_j = q^{\pm (\delta_{i,j+1} - \delta_{i,j})/2} \mathbf{X}^\pm_j \mathbf{L}_i\\
\label{Uqglncomm3}
& & \comm{\mathbf{X}^+_i}{\mathbf{X}^-_j} = 
\delta_{ij} \frac{(\mathbf{L}_{i+1}\mathbf{L}_i^{-1})^2 - (\mathbf{L}_{i+1}\mathbf{L}_i^{-1})^{-2}}{q-q^{-1}}
\eea
and, for $1 \leq i,j \leq n-1$,  the quadratic and cubic Serre relations 
\bea
\label{UqglnSerre1}
& & \comm{\mathbf{X}^\pm_i}{\mathbf{X}^\pm_j} = 0 \quad \mbox{ if } |i-j| \neq 1, \\
\label{UqglnSerre2}
& & (\mathbf{X}^\pm_i)^2 \mathbf{X}^\pm_j -
 (q+q^{-1}) \mathbf{X}^\pm_i \mathbf{X}^\pm_j \mathbf{X}^\pm_i + \mathbf{X}^\pm_j (X^\pm_i)^2 = 0 
 \quad \mbox{ if }|i-j| = 1.
\eea
Notice the replacement $q^2 \to q$ that we made in the definitions of  \cite{Burd92}.\\

\section{Results}

Before stating the results we introduce for $q,\,q^{-1} \neq 0$ and $x\in\C$ 
the symmetric $q$-number
\bel{qnumber}
[x]_q := \frac{q^x - q^{-x}}{q-q^{-1}}.
\ee
This definition extends straightforwardly to finite-dimensional matrices 
through the Taylor expansion of the exponential. 
We point out that $[-x]_q =  - [x]_q$, $[x]_{q^{-1}} =  [x]_q$
and $[x]_1 = x$. For integers one has the representation
\be
[n]_q = \sum_{k=0}^{n-1} q^{2k-n+1}
\ee
and the $q$-factorial
\bel{qfactorial}
[n]_q! := \left\{ \ba{ll} 1 & \quad n=0 \\
\prod_{k=1}^n [k]_q & \quad n \geq 1 \ea \right.
\ee
and the $q$-multinomial coefficients
\bel{partfunNM}
C_L(N) = \frac{[L]_q!}{[N]_q![L-N]_q!}, \quad
C_{L}(N,M) = \frac{[L]_q!}{[N]_q![M]_q![L-N-M]_q!}.
\ee

The first main result is a symmetry property of the generator.

\begin{theorem}
\label{Theosymmetry}
Let $H$ be the transition matrix \eref{transmatrix} of the two-component ASEP defined by 
\eref{generator} with asymmetry parameter 
$q$ and let $Y^\pm_i,\,L_j$, $i=1,2$, $j=1,2,3$ be matrices with matrix elements
\bea
\label{repLj}
& (L_1)_{\eta'\eta} = q^{- N(\eta)/2} \delta_{\eta',\eta},\,
(L_2)_{\eta'\eta} = q^{- V_\eta/2} \delta_{\eta',\eta},\,
(L_3)_{\eta'\eta} = q^{- M(\eta)/2} \delta_{\eta',\eta}  &  \\
\label{repYi}
& (Y^\pm_i)_{\eta'\eta} = \sum_{k=1}^L (Y^\pm_i(k))_{\eta'\eta} &
\eea
with
\bea
(Y^+_1(k))_{\eta'\eta} = q^{2 V_k(\eta) -V_\eta} \upsilon_k \, \delta_{\eta',\eta^{k-}} & &
(Y^-_1(k))_{\eta'\eta} = q^{-2 N_k(\eta) + N(\eta)} a_k \, \delta_{\eta',\eta^{k+}}\\
(Y^+_2(k))_{\eta'\eta} = q^{2 M_k(\eta) -M(\eta)} b_k \, \delta_{\eta',\eta^{k-}} & &
(Y^-_2(k))_{\eta'\eta} = q^{- 2 V_k(\eta) + V_\eta} \upsilon_k \, \delta_{\eta',\eta^{k+}}.
\eea
and $\eta^{k\pm}=(\gamma^k)^{\pm 1}(\eta)$ defined by \eref{flip}.
Then:\\
(a) The matrices $Y^\pm_i,\,L_j$, $i=1,2$, $j=1,2,3$ form a 
representation of the quantum algebra $U_q[\mathfrak{gl}(3)]$
\eref{Uqglndef1} - \eref{UqglnSerre2}.\\
(b) The transition matrix $H$ satisfies $\comm{H}{Y^\pm_i} = \comm{H}{L_j} = 0$ for $i=1,2$, $j=1,2,3$.
\end{theorem}

The second main result concerns reversibility. 

\begin{theorem}
\label{Theoinvmeas}
The two-component exclusion process $\eta_t$ defined by \eref{generator} with asymmetry parameter 
$q$ is reversible with the reversible measure
\bel{theo1}
\pi(\eta) = q^{\sum_{k=1}^L \left(2 k - L -1 \right)\left(a_k - b_k\right) 
+  \sum_{k=1}^{L-1} \sum_{l=1}^k \left( a_l b_{k+1}  -  b_l a_{k+1} \right)}.
\ee
\end{theorem}

\begin{remark}
(i) In terms of position variables \eref{occupos} we can write
\bel{coro1b}
\pi(\eta) = q^{\sum_{i=1}^{N(\eta)} [2x_i - L - 1 - M_{x_i}(\eta)] - 
\sum_{i=1}^{M(\eta)} [2y_i - L - 1 - N_{y_i}(\eta)]}.
\ee
(ii) In terms of conjugate occupation numbers $\bar{a}_k = 1-a_k,\,\bar{b}_k = 1-b_k$ we can 
use the identity $\sum_{k=1}^{L-1} \sum_{l=1}^k (x_{k+1}-x_l) = \sum_{k=1}^L \left(2 k - L -1 \right)x_k$ 
to write
\bel{coro1bc}
\pi(\eta) = 
q^{\sum_{k=1}^{L-1} \sum_{l=1}^k \left( \bar{a}_l \bar{b}_{k+1}  -  \bar{b}_l \bar{a}_{k+1} \right)}.
\ee
(iii)
For finite $q$ one has
\bel{positive}
\pi(\eta) > 0 \quad \forall \, \eta
\ee
so that $\pi^{-1}(\eta)$ is finite. 
\end{remark}


\section{Tools}

Here we present a review of the tools that are used to prove the theorems. Some of these
tools are not standard in the context of probability theory. The first
subsection begins with simple facts included for the benefit of readers not familiar with 
the matrix representation of properties of a Markov chain \cite{Lloy96,Schu01}. 
The second subsection
summarizes more advanced algebraic material from the theory of complete integrability of
one-dimensional quantum systems.

\subsection{Generator in matrix form}

The defining equation \eref{generator} is linear and can therefore be written in matrix form
using the transition matrix \eref{transmatrix}
\bel{generator2}
\mathcal{L} f(\eta) = - \sum_{\eta'\in\mathbb{S}^L} f(\eta') H_{\eta'\eta}.
\ee
Notice that the sum includes the term $\eta'=\eta$. 
In order to write the matrix $H$ explicitly one has to choose an concrete basis, i.e., to
each configuration $\eta$ one assigns a canonical basis vector 
and defines the ordering $\iota(\eta)$ of the basis. The set of all 
basis vectors, which we denote by $\ket{\eta}$, spans the complex
vector space $\C^{|\mathbb{S}^L|}$. 
We work with
a vector space over $\C$ rather than over $\R$ since in computations one may encounter 
eigenvectors and eigenvalues of $H$ which may be complex since $H$ is in general
not symmetric. 

Before defining a convenient ordering of the basis we
make explicit the relation between the matrix
representation \eref{transmatrix} of the generator and the definition \eref{generator} of the process
and rewrite in matrix form some of the Markov properties stated above.

\subsubsection{Matrix representation of the Markov chain}

\underline{\it Biorthogonal basis, inner product and tensor product:} In our convention the basis vectors 
$\ket{\eta}$ of dimension $d=3^L$ 
are represented as column vectors. We define also the row vector $\bra{\eta} = \ket{\eta}^T$
with the biorthogonality property
\bel{inprodbasis}
\inprod{\eta}{\eta'} = \delta_{\eta \eta'}.
\ee
The superscript $T$ on vectors or matrices denotes transposition.

Consider for arbitrary $d$ two vectors
$\bra{w}$ with components $w_i \in \C$ and $\ket{v}$ with components $\upsilon_i\in \C$.
We define the inner product by
\bel{inprod}
\inprod{w}{v} = \sum_{i=1}^{d} w_i \upsilon_i
\ee
without complex conjugation.

The tensor product 
$\ket{v}\bra{w} \equiv \ket{v}\otimes \bra{w}$ is a $d \times d$-matrix with matrix elements 
$(\ket{v}\bra{w})_{i,j}=v_i w_j$. This notation follows a convenient convention 
borrowed from quantum mechanics.
Specifically, we have the representation
\bel{unitmatrix}
\mathbf{1} = \sum_{\eta} \ket{\eta}\bra{\eta}
\ee
of the $d$-dimensional unit matrix.
We recall that for two tensor vectors
$\bra{W} = \bra{w_1} \otimes \dots \otimes \bra{w_L}$, $\ket{V} = \ket{v_1} \otimes \dots \otimes \ket{v_L}$
the inner product of factorizes:
\bel{tensorfactor}
\inprod{W}{V} = \prod_{k=1}^L \inprod{w_k}{v_k}.
\ee

\underline{\it Generator:} As a consequence of biorthogonality of the basis
one has $H_{\eta'\eta} = \bra{\eta'} H \ket{\eta}$ and therefore \eref{generator} takes the form
\bel{generator3}
\mathcal{L} f(\eta) = - \bra{f} H \ket{\eta}
\ee
where the row vector $\bra{f} = \sum_{\eta} f(\eta) \bra{\eta}$ has components $f(\eta)$.
A probability measure $P(\eta_t)\equiv P(\eta,t)$ is represented by the column vector
\bel{probvec}
\ket{P(t)} = \sum_{\eta} P(\eta,t) \ket{\eta}.
\ee
The semigroup property of the Markov chain is reflected in the time-evolution equation 
\be
\ket{P(t)} = \rme^{-Ht} \ket{P_0} 
\ee
of a probability measure $P_0(\eta)\equiv P(\eta_0)$.

\underline{\it Lowest eigenvalue and eigenvector:} Next we define the {\it summation vector}
\bel{sumvec}
\bra{s} := \sum_{\eta} \bra{\eta}
\ee 
which is the row vector where all components are equal to 1.
Normalization implies
\bel{normalization1}
\inprod{s}{P(t)} = 1 \quad \forall t\geq 0.
\ee
By taking the time-derivative one has as a consequence 
\bel{probcons}
\bra{s} H = 0
\ee
which means that the summation vector is a left eigenvector of $H$ with eigenvector 0.
This property follows from the fact that a diagonal element of $H_{\eta\eta}$ is by construction
the sum of all transition rates that appear with negative sign in the same column $\eta$ of $H$.

A stationary measure, denoted by $\ket{\pi}$,
is a right
eigenvector of $H$ with eigenvalue 0, i.e.,
\bel{statvec}
H \ket{\pi} = 0.
\ee
By the Perron-Frobenius theorem 0 is the eigenvalue of $H$ with the lowest real part.
The probability vector corresponding to a normalized stationary measure \eref{normalization1}
is denoted by $\ket{\pi^\ast}$.

For the stationary distribution we define the diagonal matrices
\bel{invmeasmatrix}
\hat{\pi}^\ast :=  \sum_{\eta} \pi^\ast(\eta)  \ket{\eta}  \bra{\eta}, \quad
\hat{\pi} :=  \sum_{\eta} \pi(\eta)  \ket{\eta}  \bra{\eta}.
\ee
For ergodic processes with finite state space one has $0< \pi^\ast(\eta) \leq 1$ for all $\eta$.
Then all powers $(\hat{\pi}^\ast)^\alpha$ exist.
In terms of this diagonal matrix we can write the generator of the reversed dynamics as 
\bel{revtrans2}
H^{rev}  = \hat{\pi}^\ast H^T (\hat{\pi}^\ast)^{-1}.
\ee
This is the matrix form of \eref{revtrans}.

\underline{\it Expectation values:} The expectation $\exval{f}_P$ of a function $f(\eta)$ 
with respect to a probability distribution $P(\eta)$
is the inner product
\be
\exval{f}_P = \inprod{f}{P} = \bra{s} \hat{f} \ket{P}
\ee
where 
\bel{functionmatrix}
\hat{f} :=  \sum_{\eta} f(\eta)  \ket{\eta}  \bra{\eta}
\ee 
is a diagonal matrix with diagonal elements $f(\eta)$. 
Notice that
\bel{felement}
f(\eta) = \bra{\eta} \hat{f} \ket{\eta} = \bra{s}  \hat{f} \ket{\eta}.
\ee

\subsubsection{The tensor basis}

In order to define a convenient ordering of the basis for $H$ we choose
the tensor approach advocated in \cite{Lloy96,Schu01} for interacting particle systems.

{\it Only one site:} We begin with the basis for a single site where $\eta\in\S$. 
For a row vector of dimension 3 with components $w_i$ we write 
$(w| = (w_1,w_2,w_3)$ and column vectors of dimension 3 with components $\upsilon_i$
we write  $|v)= (v|^T$. 
We define the inner product $(w|v) := w_1\upsilon_1 + w_2\upsilon_2 + w_3\upsilon_3$.
We choose for a single site the order $\iota_1(\eta) = 1+\eta$ and denote the corresponding
canonical basis vectors of $\C^3$
\be
|A) \equiv |1) := \left(\ba{c} 1 \\ 0 \\ 0 \ea \right), \quad
|0) \equiv |2) := \left(\ba{c} 0 \\ 1 \\ 0 \ea \right), \quad
|B) \equiv |3) := \left(\ba{c} 0 \\ 0 \\ 1 \ea \right).
\ee
With the dual basis $(\eta| = |\eta)^T$ we have the biorthogonality relation
$(\eta|\eta') = \delta_{\eta,\eta'}$
The local summation vector is given by
$(s| := (A| + (0| + (B| = (1,1,1)$.

It is useful to introduce the following matrices
with the convention $|\eta)(\eta'| \equiv |\eta)\otimes (\eta'|$:
\bea
\label{creation}
& & a^+ := |A)(0|, \quad
b^+ := |B)(0|, \quad
c^+ := |A)(B|, \\
\label{annihilation}
& & a^- := |0)(A|, \quad
b^- := |0)(B|, \quad
c^- := |B)(A|.
\eea
Having in mind the action of these operators to the right on a column vector, 
we call $a^+$ and $b^+$ creation operators, $a^-$ and $b^-$ 
are called annihilation operators and $c^\pm$ are particle exchange
operators. 

We also define the projectors
\be
\label{projection}
 \hat{a} :=  |A)(A|, \quad
\hat{\upsilon} : = |0)(0|, \quad
\hat{b} :=  |B)(B|. 
\ee
and the three-dimensional unit matrix
\bel{3dunit}
\mathds{1} = \sum_\eta |\eta)(\eta| = \hat{a} + \hat{\upsilon} + \hat{b}.
\ee
They satisfy the following relations:
\bea
\label{projhatar}
a^+ \hat{a} = b^+ \hat{a} = b^- \hat{a} = c^+ \hat{a} = 0, &\quad a^- \hat{a} = a^-, 
&\quad c^- \hat{a} = c^-\\
a^- \hat{\upsilon} = b^- \hat{\upsilon} = c^+ \hat{\upsilon} = c^- \hat{\upsilon} = 0, 
&\quad a^+ \hat{\upsilon} = a^+, &\quad b^+ \hat{\upsilon} = b^+\\
\label{projhatbr}
a^+ \hat{b} = a^- \hat{b} = b^+ \hat{b} = c^- \hat{b} = 0, &\quad b^- \hat{b} = b^-, 
&\quad c^+ \hat{b} = c^+
\eea
and
\bea
\label{projhatal}
\hat{a} a^- = \hat{a} b^+ = \hat{a} b^- = \hat{a} c^- = 0, &\quad \hat{a} a^+ = a^+, 
&\quad \hat{a} c^+ =  c^+\\
\hat{\upsilon} a^+ = \hat{a} b^+ = \hat{\upsilon} c^+ = \hat{\upsilon} c^- = 0, 
&\quad \hat{\upsilon} a^- = a^-, &\quad \hat{\upsilon} b^- =  b^-\\
\label{projhatbl}
\hat{b} a^+ = \hat{b} a^- = \hat{b} b^- = \hat{b} c^+ = 0, &\quad \hat{b} b^+ = b^+, 
&\quad \hat{b} c^- =  c^-. 
\eea
%
%
With the occupation variables \eref{lonv} for a single site we have the projector properties
\bel{proj}
\hat{a} |\eta) = a |\eta), \quad \hat{b} |\eta) = b |\eta), \quad \hat{\upsilon} |\eta) = v |\eta)
\ee
for $\eta \in \{A,0,B\}$ and $a\equiv a(\eta)=\delta_{\eta,1}$ and so on.
Moreover,
\bea
\label{summ+}
& & (s| a^+ = (s| \hat{\upsilon}, \quad (s| b^+ = (s| \hat{\upsilon}, \quad (s| c^+ = (s| \hat{b},\\
\label{summ-}
& & (s| a^- = (s| \hat{a}, \quad (s| b^- = (s| \hat{b}, \quad (s| c^- = (s| \hat{a}.
\eea

{\it $L$ sites:} For a configuration $\eta\in \S^L$ it is natural and indeed convenient
to choose the ternary ordering $\iota_L(\eta) = 1+ \sum_{k=1}^L \eta(k) 3^{k-1}$ of the
basis vectors. This corresponds to
the tensor basis defined by
\bel{tensorbasis}
\ket{\eta} := |\eta_1) \otimes |\eta_2) \otimes \dots \otimes |\eta_L), \quad
\bra{\eta} := (\eta_1| \otimes (\eta_2| \otimes \dots \otimes (\eta_L|
\ee
spanning the vector space $(\C^{3})^{\otimes L}$ of dimension $3^L$.
We shall also use the notations $\ket{\bfz}$ and $\ket{\bfx,\bfy}$ instead of  $\ket{\eta}$. 
Furthermore, if a configuration $\eta$ has $N$ particles
of type $A$ and $M$ particles of type $B$ we may denote this fact by writing 
$\ket{\eta_{N,M}}$ for the corresponding basis vector.

The summation vector $\bra{s}$ is given by
\be
\bra{s} := (s|^{\otimes L}.
\ee
This is the row vector of dimension $3^L$ where all components are equal to 1.
The summation vector restricted to configurations with a fixed number
$N$ of particles of type $A$ and $M$ particles of type $B$ is denoted by 
\bel{sumvecNM}
\bra{s_{N,M}} = \sum_{\eta_{N,M}} \bra{\eta_{N,M}}.
\ee
The basis vector $\ket{\eta_{0,0}}$ corresponding to the empty lattice
is denoted by $\ket{\emptyset} = \bra{s_{0,0}}^T$.

For matrices $M$ the expression $M^{\otimes j}$ will denote the $j$-fold tensor
product of $M$ withself if $j>1$. For $j=1$ we define $M^{\otimes 1} := M$ and for $j=0$ we define 
$M^{\otimes 0} = 1$
with the $c$-number $1$.
For arbitrary $3\times 3$-matrices $u$
we define tensor operators
\be
u_k := \mathds{1}^{\otimes (k-1)} \otimes u \otimes \mathds{1}^{\otimes (L-k)}.
\ee
Multilinearity of the tensor product yields
$u_k v_{k+1} = \mathds{1}^{\otimes (k-1)} \otimes [(u \otimes \mathds{1}) (\mathds{1} \otimes v)]
\otimes \mathds{1}^{\otimes (L-k-1)} = \mathds{1}^{\otimes (k-1)} \otimes ( u \otimes v )
\otimes \mathds{1}^{\otimes (L-k-1)}$ and the commutator property
$u_k v_l = v_l u_k$ for $k\neq l$.

For $u=\hat{a}$ or $\hat{b}$ we note the projector lemma \cite{Beli15} which will be used repeatedly below.
\begin{lemma}
\label{projlem}
The tensor occupation operators $\hat{a}_k$, $\hat{b}_k$ act as projectors
\be
\label{projabL}
\hat{a}_k \ket{\eta} = a_k \ket{\eta} = \sum_{i=1}^{N(\eta)} \delta_{x_i,k} \ket{\eta}, \quad
\hat{b}_k \ket{\eta} = b_k \ket{\eta} = \sum_{i=1}^{M(\eta)} \delta_{y_i,k} \ket{\eta}
\ee
with the occupation variables $a_k$ and $b_k$ \eref{lonv} (or particle coordinates $x_i$ and $y_i$ respectively)
understood as functions of $\eta$ or $\bfz=\eta$.
\end{lemma}

\begin{remark}
One obtains the diagonal matrix $\hat{f}$ \eref{functionmatrix} of a
function $f(\eta)$ by the following simple recipe:
In $f(\eta)$ one substitutes $\eta_i$ by the diagonal matrix $\hat{\eta}_i$ where 
$\hat{\eta}_i = \hat{a}_i$ if $\eta_i = 0$, $\hat{\eta}_i = \hat{\upsilon}_i$ if $\eta_i = 1$, and
$\hat{\eta}_i = \hat{b}_i$ if $\eta_i = 2$.
\end{remark}

\subsubsection{Construction of the generator in the tensor basis}

Consider first $L=2$ and denote the transition matrix \eref{transmatrix} for two sites by $h$.
From the definition of the process one computes the off-diagonal part
by observing that $\ket{0A} = \ket{\pi^{12}(\{A0\})}  = (a^- |A)) \otimes (a^+ |0))
= (a^- \otimes a^+) (|A) \otimes |0)) = a_1^- a_2^+ \ket{A0}$ and therefore
$h_{\{0A\}\{A0\}} = qw \bra{0A} a_1^- a_2^+ \ket{A0}$ for the transition $A0\to 0A$,
and so on. The corresponding diagonal elements $h_{\eta\eta}$ follow from \eref{localrates}
with the substitution of the occupation variables by the respective projectors according to 
Lemma \eref{projlem}.
With the tensor basis \eref{tensorbasis} for two sites one thus obtains
the $9\times9$-matrix 
\bel{hoppingmatrix}
h = w \left( \ba{ccccccccc}
\,\, 0 \,\, &\,\, 0 \,\, &\,\, 0 \,\, &\,\, 0 \,\, &\,\, 0 \,\, &\,\, 0 \,\, &\,\, 0 \,\, &\,\, 0 \,\, &\,\, 0 \,\,  \\
0 & q & 0 & -q^{-1} & 0 & 0 & 0 & 0 & 0 \\
0 & 0 & q & 0 & 0 & 0 & -q^{-1} & 0 & 0 \\
0 & -q & 0 & q^{-1} & 0 & 0 & 0 & 0 & 0 \\
0 & 0 & 0 & 0 & 0 & 0 & 0 & 0 & 0 \\
0 & 0 & 0 & 0 & 0 & q & 0 & -q^{-1} & 0 \\
0 & 0 & -q & 0 & 0 & 0 & q^{-1} & 0 & 0 \\
0 & 0 & 0 & 0 & 0 & -q & 0 & q^{-1}  & 0 \\
0 & 0 & 0 & 0 & 0 & 0 &0 & 0 & 0 \\
\ea \right).
\ee

Next we embed this process on two neighbouring sites $(k,k+1)$ in $\Lambda$.
By the multilinearity of the tensor product the generator becomes
$h_{k,k+1} = \mathds{1}^{ \otimes (k-1)}  \otimes h  \otimes \mathds{1}^{ \otimes (L-k-1)}$
where the hopping matrices
\bea
h_{k,k+1} & : = & w q \left( \hat{a}_k \hat{\upsilon}_{k+1} - a^-_k a^+_{k+1}  + 
\hat{\upsilon}_k \hat{b}_{k+1} - b^+_k b^-_{k+1} 
+ \hat{a}_k \hat{b}_{k+1} - c^-_k c^+_{k+1}  \right) \nonumber \\
\label{hoppingmatrixk}
& & +
w q^{-1}  \left(\hat{\upsilon}_k \hat{a}_{k+1} - a^+_k a^-_{k+1}  + 
\hat{b}_k \hat{\upsilon}_{k+1} - b^-_k b^+_{k+1} 
+  \hat{b}_k \hat{a}_{k+1} - c^+_k c^-_{k+1}  \right)
\eea
act non-trivially only on the subspace corresponding to sites $k,k+1$ in the tensor space.
The off-diagonal elements of $h_{k,k+1}$ are the transition rates $h_{\eta^{kk+1}\eta}$ in the tensor 
basis \eref{tensorbasis}. The diagonal elements of $h_{k,k+1}$ defined in \eref{transmatrix}
follow from probability conservation.
The generator for the two-component ASEP on the lattice $\{1,\dots,L\}$ then follows as
\bel{2ASEPgen}
H = \sum_{k=1}^{L-1} h_{k,k+1}.
\ee

As will be seen below, the generator $H$ is closely related to the Hamiltonian operator
of a quantum spin chain.
This is true for other interacting particle systems and motivates the
terms
``Quantum spin techniques'' \cite{Lloy96} or
``quantum Hamiltonian formalism'' \cite{Schu01} for the representation of the generator
of an interacting particle system in the tensor basis.

\begin{remark}
Because of multilinearity \eref{summ+}, \eref{summ-} are lifted to $\bra{s} a^+_k =  \bra{s} \hat{\upsilon_k}$, 
$\bra{s} b^+_k =  \bra{s} \hat{\upsilon}_k$, $\bra{s} c^+_k =  \bra{s} \hat{b}_k$ and
$\bra{s} a^-_k =  \bra{s} \hat{a}_k$, $\bra{s} b^-_k =  \bra{s} \hat{b}_k$, $\bra{s}c^-_k =  \bra{s} \hat{a}_k$.
This yields $\bra{s} h_{k,k+1} = 0$ which implies probability conservation \eref{probcons} .
\end{remark}

\subsection{Quantum algebra $U_q[\mathfrak{gl}(3)]$ and the Perk-Schultz quantum chain}
\label{gl3q}

Above we have defined the quantum algebra $U_q[\mathfrak{gl}(3)]$ \eref{Uqglndef1} -
\eref{UqglnSerre2}. It is convenient to work also with the subalgebra $U_q[\mathfrak{sl}(3)]$.

\subsubsection{Relation between $U_q[\mathfrak{gl}(n)]$ and $U_q[\mathfrak{sl}(n)]$}

We introduce generators $\tilde{\mathbf{H}}_i$, $1\leq i \leq n$ and $\mathbf{H}_i$ $1\leq i \leq n-1$ through
\bel{complement}
q^{-\tilde{\mathbf{H}}_i/2} = \mathbf{L}_i, \quad \mathbf{H}_i = 
\tilde{\mathbf{H}}_i - \tilde{\mathbf{H}}_{i+1}.
\ee
Then the quantum algebra $U_q[\mathfrak{sl}(n)]$ is the subalgebra generated by 
$q^{\pm \mathbf{H}_i/2}$, and $\mathbf{X}^\pm_i$, $i=1,\dots,n-1$ with relations
\eref{UqglnSerre1}, \eref{UqglnSerre2}
and
\bea
& & q^{\mathbf{H}_i/2} q^{-\mathbf{H}_i/2} = q^{-\mathbf{H}_i/2} q^{\mathbf{H}_i/2} = I\\
\label{Uqslncomm1b}
& & q^{\mathbf{H}_i/2} q^{\mathbf{H}_j/2} = q^{\mathbf{H}_j/2} q^{\mathbf{H}_i/2} \\
\label{Uqslncomm2b}
& & q^{\mathbf{H}_i} \mathbf{X}^\pm_j q^{-\mathbf{H}_i} = q^{\pm A_{ij}} \mathbf{X}^\pm_j\\
\label{Uqslncomm3b}
& & \comm{\mathbf{X}^+_i}{\mathbf{X}^-_j} = \delta_{ij} [\mathbf{H}_i]_q.
\eea
with the unit $I$ and the Cartan matrix
\bel{basicdefs}
 A_{ij} :=  \left\{ \ba{rl} 
2 & \mbox{ if } i=j \\ 
-1 & \mbox{ if } j = i\pm 1\\ 
0 & \mbox{else.} \ea \right.
\ee
of simple Lie algebras of type $A_n$.
That $U_q[\mathfrak{sl}(n)]$ is a subalgebra of $U_q[\mathfrak{gl}(n)]$ can be seen by noticing that 
$\sum_{i=1}^n \tilde{\mathbf{H}}_i$ belongs to the center of $U_q[\mathfrak{gl}(n)]$ \cite{Jimb86}.

We remark that the commutation relations \eref{Uqslncomm1b}, \eref{Uqslncomm2b} can be 
substituted by
\bea
\label{Uqslncomm1}
& & \comm{\mathbf{H}_i}{\mathbf{H}_j} = 0 \\
\label{Uqslncomm2}
& & \comm{\mathbf{H}_i}{\mathbf{X}^\pm_j} = \pm A_{ij} \mathbf{X}^\pm_j
\eea
and defining $q^{\pm \mathbf{H}_i/2}$ as a formal series 
through the Taylor expansion of the exponential.

\subsubsection{Finite-dimensional representations for $n=3$}

For $n=3$ the quadratic Serre relations 
\eref{UqglnSerre1} are void.
By introducing \cite{Burr90}
\bel{Def:X3pm}
\mathbf{X}^\pm_3 := q^{1/2} \mathbf{X}^\pm_1 \mathbf{X}^\pm_2 - q^{-1/2} \mathbf{X}^\pm_2 \mathbf{X}^\pm_1
\ee
the cubic Serre relations reduce to quadratic relations. One has instead of \eref{UqglnSerre2}
\bea
& & 
q^{-1/2} \mathbf{X}^\pm_1 \mathbf{X}^\pm_3 - q^{1/2} \mathbf{X}^\pm_3 \mathbf{X}^\pm_1 = 0\\
& & 
q^{1/2} \mathbf{X}^\pm_2 \mathbf{X}^\pm_3 - q^{-1/2} \mathbf{X}^\pm_3 \mathbf{X}^\pm_2 = 0.
\eea
and also
\be 
\label{Uqsu3comm2a}
\comm{\mathbf{H}_i}{\mathbf{X}^\pm_3} = \pm  \mathbf{X}^\pm_3. 
\ee

In order to distinguish the matrices corresponding to  the
three-dimensional fundamental representation 
from the abstract generators we use lower case letters. 
In terms of \eref{creation} - \eref{projection}
the three-dimensional fundamental representation of $U_q[\mathfrak{gl}(3)]$ is given by:
\bea
\label{fundrepX}
& & x_1^\pm = a^\pm, \quad x_2^\pm = b^\mp \\
\label{fundrepH3}
& & \tilde{h}_1 = \hat{a}, \quad \tilde{h}_2 = \hat{\upsilon}, \quad \tilde{h}_3 = \hat{b}, 
\eea
corresponding to
\be 
\label{fundrepH}
h_1 = \hat{a} - \hat{\upsilon}, \quad
h_2 = \hat{\upsilon} - \hat{b}.
\ee
for the representation of the generators $H_i$ of $U_q[\mathfrak{sl}(3)]$.
We also mention the representation $x^\pm_3 = \pm q^{\pm 1/2} c^\pm$.

Next we introduce the coproduct 
\bea
\label{coprod1}
\Delta(\mathbf{X}^\pm_i) & = & \mathbf{X}^\pm_i \otimes q^{\mathbf{H}_i/2} + 
q^{-\mathbf{H}_i/2} \otimes \mathbf{X}^\pm_i \\
\label{coprod2}
\Delta(\mathbf{H}_i) & = & \mathbf{H}_i \otimes \mathds{1} + \mathds{1} \otimes \mathbf{H}_i.
\eea
By repeatedly applying the coproduct to the fundamental representation, 
we construct the tensor representations of $U_q[\mathfrak{sl}(3)]$, denoted by capital letters,
\bel{repX}
X^\pm_i = \sum_{k=1}^L  X^\pm_i(k), \quad H_i = \sum_{k=1}^L  H_i(k)
\ee
with
\bea
X^\pm_i(k) & = & q^{-H_i/2} \otimes \dots \otimes q^{-H_i/2} \otimes X^\pm_i \otimes 
q^{H_i/2} \dots \otimes q^{H_i/2}, \\
H_i(k) & = & \mathds{1} \otimes \dots \otimes \mathds{1} \otimes H_i \otimes 
\mathds{1} \dots \otimes \mathds{1}.
\eea
For the full quantum algebra $U_q[\mathfrak{gl}(3)]$ we have
\bel{repH3}
\tilde{H}_1  = \sum_{k=1}^L \hat{a}_k =: \hat{N}, 
\quad \tilde{H}_2 = \sum_{k=1}^L \hat{\upsilon}_k=: \hat{V}, 
\quad \tilde{H}_3 = \sum_{k=1}^L \hat{b}_k =: \hat{M}.
\ee
Here $\hat{N}$ and $\hat{M}$ are the particle number operators satisfying
\be 
\hat{N} \ket{\eta_{N,M}} = N \ket{\eta_{N,M}}, \quad \hat{M} \ket{\eta_{N,M}} = M \ket{\eta_{N,M}}.
\ee
The unit $I$ is represented by the $3^L$-dimensional unit matrix
$\mathbf{1} := \mathds{1}^{\otimes L} = \hat{N} + \hat{V} + \hat{M}$.

Notice that $H_1(k) = \hat{a}_k - \hat{\upsilon}_k$ and $H_2(k) = \hat{\upsilon}_k - \hat{b}_k$.
In the $L$-fold tensor product $X^\pm_i(k)$ $(H_i(k))$ the term $X^\pm_i$ ($H_i$) is the $k$th factor.
Therefore \eref{fundrepX} yields
\bea
\label{repX1k}
X^\pm_1(k) & = & q^{- \half \sum_{j=1}^{k-1} (\hat{a}_j - \hat{\upsilon}_j) + 
\half \sum_{j=k+1}^L (\hat{a}_j - \hat{\upsilon}_j)} a^{\pm}_k \\
\label{repX2k}
X^\pm_2(k) & = & q^{- \half\sum_{j=1}^{k-1} (\hat{\upsilon}_j - \hat{b}_j) + 
\half \sum_{j=k+1}^L (\hat{\upsilon}_j - \hat{b}_j)} b^{\mp}_k .
\eea
One has $(X_i^\pm(k))^2=0$, $X_i^\pm(k) X_j^\mp(k)=0$ for $i\neq j$ and
\bea
X_i^\pm(k) X_i^\pm(l) & = & \left\{ \ba{cl} 
q^{\pm 2} X_i^\pm(l) X_i^\pm(k) & k<l \\ 
0 & k=l \\ q^{\mp 2} X_i^\pm(l) X_i^\pm(k) & k>l 
\ea \right. \\
X_i^\pm(k) X_j^\mp(l) & = & X_i^\pm(l) X_j^\mp(k) \mbox{ for } i\neq j.
\eea
Thus the spatial order in which particles are created (or annihilated)
by applying the operators $X_i^\pm(k)$ gives rise to
combinatorial issues when building many-particle configurations 
from the reference state corresponding to the empty lattice.

\subsubsection{The Perk-Schultz quantum spin chain}

We introduce the Perk-Schultz quantum spin chain \cite{Perk81}
\bel{PSHam}
G = \sum_{k=1}^{L-1} g_{k,k+1}
\ee
where $g_{k,k+1}$ is
reminiscent of \eref{hoppingmatrixk}, but with all non-zero
off-diagonal elements equal to $-1$, i.e.,
\bea
\frac{1}{w} g_{k,k+1} & = & q \left( \hat{a}_k \hat{\upsilon}_{k+1} + 
\hat{\upsilon}_k \hat{b}_{k+1} + \hat{a}_k \hat{b}_{k+1} \right)
+ q^{-1}  \left(\hat{\upsilon}_k \hat{a}_{k+1} + 
\hat{b}_k \hat{\upsilon}_{k+1}+  \hat{b}_k \hat{a}_{k+1}  \right)  \nonumber \\
& & 
- a^-_k a^+_{k+1}   - b^+_k b^-_{k+1}  - c^-_k c^+_{k+1} 
- a^+_k a^-_{k+1}   - b^-_k b^+_{k+1}  - c^+_k c^-_{k+1}.
\eea
The $g_{k,k+1}$ satisfy the
defining relations of the $(3,0)$-quotient of the Hecke algebra 
\cite{Mart92} which then implies \cite{Degu89}
that 
\bel{symmetryG}
\comm{g_{k,k+1}}{X^\pm_i} = \comm{g_{k,k+1}}{H_i} = 0.
\ee
Thus
the Perk-Schultz quantum Hamiltonian is symmetric under the action of the 
quantum algebra $U_q[\mathfrak{sl}(3)]$ and then trivially also under
$U_q[\mathfrak{gl}(3)]$.

\section{Proofs}

It was pointed out in \cite{Alca93} that the hopping matrices $h_{k,k+1}$  \eref{hoppingmatrixk} for the 
ASEP with second-class particles satisfy the
defining relations of the same $(3,0)$-quotient of the Hecke algebra as the $g_{k,k+1}$
of the Perk-Schultz quantum chain. This fact implies the existence of representation matrices
of the generators of $U_q[\mathfrak{gl}(3)]$ with which the hopping matrices $h_{k,k+1}$
and hence the generator \eref{2ASEPgen} commutes.
However, in order to make this symmetry property useful for probabilistic and physical applications
one must solve the main problem that was left open in \cite{Alca93}, which is to actually
construct this representation. This is the content of Theorem \eref{Theosymmetry}, proved below.
It turns out that both Theorem \eref{Theosymmetry} and Theorem \eref{Theoinvmeas}
are consequences of a proposition that we first motivate and then prove.

\subsection{Perk-Schultz chain and ASEP with second-class particles}

The point in case is that $G$ and $H$ differ only by multiplicative factors $q$ and $q^{-1}$
in their off-diagonal elements. 
Therefore the following proposition is a natural working hypothesis.

\begin{proposition}
\label{Propdiag}
Let $H$ and $G$ be the matrices defined in \eref{2ASEPgen} and \eref{PSHam}. 
There exists a diagonal similarity transformation $R$ such that 
\bel{trafo}
G = R^{-1} H R
\ee  
with an invertible matrix $R$ of dimension $3^L$.
\end{proposition}

\begin{proof} 

In order to prove this proposition we use the quantum algebra symmetry of the Perk-Schultz chain
to first construct a good candidate for such
a transformation and then prove that it satisfies \eref{trafo}.\\

\noindent \underline{\it (1) Construction of a candidate $R$:} 
Fix the numbers $N$ of particles of type $A$ and $M$ of particles of type $B$.
By ergodicity the summation vector $\bra{s_{N,M}}$ \eref{sumvecNM}
is the unique left eigenvector with eigenvalue 0 of $H$ restricted to configurations with 
$N$ particles of type $A$ and $M$ particles of type $B$.

For $N=M=0$ one readily verifies that
$\bra{s_{0,0}} H = \bra{s_{0,0}} G = 0$.
The symmetry \eref{symmetryG} then yields
$\bra{s_{0,0}} (X_1^-)^N (X_2^+)^M G = 0$. Therefore, if $R$ exists it has the property
\be
\bra{s_{N,M}} = Y_L(N,M)  \bra{0,0} (X_1^-)^N (X_2^+)^M R^{-1}
\ee
with a normalization factor $Y_L(N,M)$.
Notice now that $\bra{s} = \sum_{N=0}^L \sum_{M=0}^{L-N} \bra{s_{N,M}}$. Therefore
\bea
\bra{s} R & =  & \sum_{N=0}^L \sum_{M=0}^{L-N} \bra{s_{N,M}} R \\
& = &
\label{D1}
 \sum_{N=0}^L \sum_{M=0}^{L-N} Y_L(N,M)  \bra{s_{0,0}} (X_1^-)^N (X_2^+)^M.
\eea

Now we make the diagonal ansatz $R = \sum_{\eta} R(\eta) \ket{\eta}\bra{\eta}$,
normalized such that
\bel{normD}
\bra{s_{0,0}} R \ket{\emptyset} = 1 = \bra{s_{0,0}} R^{-1} \ket{\emptyset}.
\ee 
Thus we obtain
from \eref{D1} that
\bel{D2}
R(\eta_{N,M}) = Y_L(N,M)  \bra{s_{0,0}} (X_1^-)^N (X_2^+)^M \ket{\eta_{N,M}}.
\ee
The normalization factors $Y_L(N,M)$ are arbitrary, since they can be absorbed by 
redefining $R = \tilde{R} E$
where $E$ is a diagonal matrix with matrix elements $Y_L(N,M)$ in the block $N,M$. Since
both $G$ and $H$ conserve particle number for both species one has $EHE^{-1} = H$ and
$EGE^{-1} = G$. Therefore $G = \tilde{R}^{-1} H \tilde{R}$ which implies that 
$G = \tilde{R}^{-1} H \tilde{R}$. Hence
the $Y_L(N,M)$ can indeed be chosen arbitrarily. 
It turns out to be convenient to choose
\be
Y_L(N,M) = \left( [N]_q ! [M]_q ! \right)^{-1}.
\ee

This reduces the computation of the diagonal elements of $R$ to the computation 
of the matrix elements  $\bra{0,0} (X_1^-)^N (X_2^+)^M \ket{\eta_{N,M}}$ from the explicit
representation \eref{repX}.  
In order to compute $R$ we first set $M=0$ and use 
\be
\bra{s_{0,0}} \frac{(X^-_1)^N}{[N]_q!} = \sum_{\bfx} 
q^{\sum_{k=1}^N x_k - N \frac{L+1}{2}} \bra{\bfx,\emptyset} 
\ee
which is a simple adaptation of an analogous result for the standard
single-species ASEP taken from \cite{Schu97}.
The sum over $\bfx$ stands for the sum over all particle positions $x_i$
ordered such that $x_i < x_j$ for $i<j$, which is the sum over all distinct sets
of particle positions. Therefore $\sum_{\bfx,\emptyset}\bra{\bfx,\emptyset} 
= \bra{s_{N,0}}$ which allows us to write
\bel{vacuumactionX1}
R(\bfx,\emptyset) = q^{\sum_{k=1}^{|\bfx|} x_k - |\bfx| \frac{L+1}{2}}
\ee
and, using \eref{projabL},
\be
\bra{s_{0,0}} \frac{(X^-_1)^N}{[N]_q!} = 
\bra{s_{N,0}} q^{\sum_{k=1}^L \left(k  - \frac{L+1}{2} \right) \hat{a}_k} .
\ee

Next we apply $(X^+_2)^M$ \eref{repX2k} to this vector
and observe that for the factors $q^{\pm H_2(k)/2}$ 
that appear in $X^+_2$ (instead of the $q^{\pm H_1(k)/2}$ that appear in $X^-_1$) any 
$A$-particle is like a non-existent site, since $H_2$ is build from projectors
on $B$-particles and vacancies. Hence, with regard to the action of $(X^+_2)^M$,
the existence of $A$-particles in a configuration
$\eta_{N,M}\equiv z_{N,M}$ with $N$ particles of type $A$ 
behaves like a reduction of system size $L\to \tilde{L} = L-N$ and a coordinate shift 
\be
y_i \to \tilde{y}_i = y_i - N_{y_i}(\bfz)
\ee 
for $B$-particles with $N_{k}(\bfz)$ defined in \eref{NyMx}.
Therefore the action of $(X^+_2)^M/[M]_q!$ on $\bra{s_{0,0}} \frac{(X^-_1)^N}{[N]_q!}$
yields a $q$-factor similar to the one in \eref{vacuumactionX1}, but with $N$ replaced by $M$,
$L$ replaced by $\tilde{L}$, $x_k$ replaced by $\tilde{y}_k$ and $q$ replaced by $q^{-1}$.
We conclude
\be
\bra{s_{0,0}} \frac{(X^-_1)^N}{[N]_q!} \frac{(X^+_2)^M}{[M]_q!} =  \sum_{\bfz_{N,M}}   
R(\bfz_{N,M})  \bra{\bfz_{N,M}} 
\ee
where the sum is over all distinct coordinate sets and
\bel{Rbfz}
R(\bfz_{N,M}) = q^{\frac{1}{2} \left[ (M-N)(L+1)-MN \right] + \sum_{i=1}^N x_i - \sum_{i=1}^M \tilde{y}_i}.
\ee

Next observe that for any configuration $\eta$
\bea
\sum_{i=1}^{M(\eta)} N_{y_i}(\eta) & = & \sum_{k=1}^L \sum_{l=1}^{k-1} a_l b_k 
= N(\eta) M(\eta) - \sum_{k=1}^L \sum_{l=1}^{k-1} b_l a_k \nonumber \\
& = & N(\eta) M(\eta) - \sum_{i=1}^{N(\eta)} M_{x_i}(\eta).
\eea
For $\eta=\bfz_{N,M}$ this yields 
\be
R(\bfz_{N,M}) = q^{\frac{1}{2} \left[ \sum_{i=1}^N (2 x_i - L-1 - M_{x_i}(\eta))
- \sum_{i=1}^M (2 y_i - L-1 - N_{y_i}(\eta)) \right]}.
\ee
and with Lemma \eref{projlem} 
\be
\bra{s_{0,0}} \frac{(X^-_1)^N}{[N]_q!} \frac{(X^+_2)^M}{[M]_q!} = \bra{s_{N,M}} q^{\frac{1}{2}\hat{U}}
\ee
with
\be
\hat{U} =  \sum_{k=1}^L \left(2 k - L -1 \right)\left(\hat{a}_k - \hat{b}_k\right) 
+  \sum_{k=1}^{L-1} \sum_{l=1}^k \left( \hat{a}_l \hat{b}_{k+1}  -  \hat{b}_l \hat{a}_{k+1} \right).
\ee
Notice that the matrix $\hat{U}$ does not depend on $N$ and $M$.
Taking the sum over $N$ and $M$ then yields from \eref{D1}
the diagonal candidate matrix
\be
R = q^{\frac{1}{2}\hat{U}} .
\ee

\noindent \underline{\it (2) Proof of the transformation property \eref{trafo}:} 
We stress that the properties of $R$ that we have used in its construction are only
necessary conditions for the transformation property \eref{trafo} to be valid.
In order to prove this property we need two more technical 
ingredients. The first is a transformation lemma, proved in \cite{Beli15}

\begin{lemma}
\label{abtrafolmx}
For any finite $p\neq 0$ we have 
\bea
\label{abtrafo1a}
& & p^{\hat{a}_l} a_x^\pm p^{-\hat{a}_l} = p^{\pm \delta_{l,x}} a_x^\pm, \quad
p^{\hat{b}_l} a_x^\pm p^{-\hat{b}_l} =  a_x^\pm, \\
\label{abtrafo1b}
& & p^{\hat{b}_l} b_x^\pm p^{-\hat{b}_l} = p^{\pm \delta_{l,x} } b_x^\pm, \quad
p^{\hat{a}_l} b_x^\pm p^{-\hat{a}_l} =  b_x^\pm \\
\label{abtrafo2a}
& & p^{\hat{a}_l \hat{b}_m} a_x^\pm p^{-\hat{a}_l \hat{b}_m} = p^{\pm \delta_{l,x} \hat{b}_m} a_x^\pm,\\
\label{abtrafo2b}
& & p^{\hat{a}_l \hat{b}_m} b_x^\pm p^{-\hat{a}_l \hat{b}_m} = p^{\pm \delta_{m,x} \hat{a}_l} b_x^\pm.
\eea
\end{lemma}

Applying these transformations yields the following auxiliary result.

\begin{lemma}
\label{abctrafo}
The local creation and annihilation operators transform as follows:
\bea
\label{axtrafo}
R a_x^\pm R^{-1} & = & q^{\mp \half \sum_{k=1}^{x-1} (\hat{b}_k-\mathds{1}) 
\pm \half \sum_{k=x+1}^{L} (\hat{b}_k-\mathds{1})} a_x^\pm, \\
\label{bxtrafo}
R b_x^\pm R^{-1} & = & q^{\pm \half \sum_{k=1}^{x-1} (\hat{a}_k-\mathds{1}) 
\mp \half \sum_{k=x+1}^{L} (\hat{a}_k-\mathds{1})} b_x^\pm \\
\label{cxtrafo}
R c_x^\pm R^{-1} & = & q^{\pm \half \sum_{k=1}^{x-1} (\hat{\upsilon}_k+\mathds{1}) 
\mp \half \sum_{k=x+1}^{L} (\hat{\upsilon}_k+\mathds{1})} c_x^\pm
\eea
\end{lemma}

\begin{proof}
To prove these identities we use Lemma \eref{abtrafolmx} and commutativity of the
projection operators. This yields
\bea
& & p^{\sum_{k=1}^{L-1} \sum_{l=1}^k (\hat{a}_l \hat{b}_{k+1} - \hat{b}_l \hat{a}_{k+1})}
a_x^\pm p^{-\sum_{k=1}^{L-1} \sum_{l=1}^k (\hat{a}_l \hat{b}_{k+1} - \hat{b}_l \hat{a}_{k+1})}  
= \nonumber \\
& & p^{\sum_{k=1}^{L-1} \sum_{l=1}^k (\pm \delta_{l,x} \hat{b}_{k+1} \mp \delta_{k+1,x} \hat{b}_l )} 
a_x^\pm =  p^{ \mp \sum_{k=1}^{x-1} \hat{b}_{k} \pm \sum_{k=x+1}^{L} \hat{b}_k )} a_x^\pm 
\eea
and
\be 
p^{\sum_{k=1}^{L} (2k-L-1) (\hat{a}_k - \hat{b}_{k}) }
a_x^\pm p^{- \sum_{k=1}^{L} (2k-L-1) (\hat{a}_k - \hat{b}_{k})} =
p^{\pm (2x-L-1)} a_x^\pm
\ee
Joining both yields \eref{axtrafo} and a similar computation yields \eref{bxtrafo}.
Finally, \eref{cxtrafo} follows from $c_k^+ = a_k^+b_k^-$, $c_k^- = b_k^+a_k^-$. \qed
\end{proof}

Now we can prove \eref{trafo}.
We split $g_{k,k+1} = g^\rmd_{k,k+1} - g^{\rm o}_{k,k+1}$ 
into its diagonal part $g^\rmd_{k,k+1}$ and its off-diagonal part $g^{\rm o}_{k,k+1}$ 
and similarly for $h_{k,k+1}$. Trivially one has $R g^\rmd_{k,k+1} R^{-1} =
 g^{\rmd}_{k,k+1} =  h^{\rmd}_{k,k+1}$.

For the offdiagonal parts, we consider first $a_k^\pm a_{k+1}^\mp$. 
Eq. \eref{axtrafo} in Lemma \eref{abctrafo} yields
\be 
R a_k^\pm a_{k+1}^\mp R^{-1} = 
a_k^\pm q^{\pm \half  (\hat{b}_{k+1}-\mathds{1})} 
q^{\pm \half (\hat{b}_k-\mathds{1})} a_{k+1}^\mp.
\ee 
The general projector property $p^{\hat{b}_{m}} = 1 + (p-1) \hat{b}_{m}$.
together with \eref{projhatbr} and \eref{projhatbl} applied to the subspaces $k$ and $k+1$
lead to 
\be 
R a_k^\pm a_{k+1}^\mp R^{-1} = q^{\mp 1} a_k^\pm a_{k+1}^\mp.
\ee
In the same fashion one proves 
\be
R b_k^\pm b_{k+1}^\mp R^{-1} = q^{\pm 1} b_k^\pm b_{k+1}^\mp.
\ee

Finally, by similar arguments 
\be
R c_k^\pm c_{k+1}^\mp R^{-1} = c_k^\pm q^{\mp \half  (\hat{\upsilon}_{k+1}+\mathds{1})} 
q^{\mp \half (\hat{\upsilon}_k+\mathds{1})} c_{k+1}^\mp \nonumber \\
= q^{\mp 1} c_k^\pm c_{k+1}^\mp.
\ee 
Comparing with \eref{hoppingmatrixk} shows that $R g^{\rm o}_{k,k+1} R^{-1}
= h^{\rm o}_{k,k+1}$ and thus $R G R^{-1} = H$. 
\qed
\end{proof}

\subsection{Proof of Theorem \eref{Theosymmetry}}

\begin{proof} From 
\eref{axtrafo}, the commutativity of the projectors at different sites,
and \eref{fundrepH}
lifted to the tensor space, one finds that
the local generators $X^\pm_i(r)$ transform as follows:
\bea
\label{X1ptrafo}
Y^+_1(r) := R X^+_1(r) R^{-1} & = & q^{ \sum_{k=1}^{r-1} \hat{\upsilon}_k
- \sum_{k=r+1}^{L} \hat{\upsilon}_k} a_r^+, \\
\label{X1mtrafo}
Y^-_1(r) := R X^-_1(r) R^{-1} & = & q^{- \sum_{k=1}^{r-1} \hat{a}_k
+ \sum_{k=r+1}^{L} \hat{a}_k} a_r^-, \\
\label{X2ptrafo}
Y^+_2(r) := R X^+_2(r) R^{-1} & = & q^{ \sum_{k=1}^{r-1} \hat{b}_k 
- \sum_{k=r+1}^{L} \hat{b}_k} b_r^- \\
\label{X2mtrafo}
Y^-_2(r) := R X^-(r) R^{-1} & = & q^{- \sum_{k=1}^{r-1} \hat{\upsilon}_k 
+ \sum_{k=r+1}^{L} \hat{\upsilon}_k} b_r^+
\eea
Moreover, since $R$ and $\tilde{H}_i$ are all diagonal one has
$ R \tilde{H}_i R^{-1} = \tilde{H}_i$. 
Commutativity of the hopping matrices $h_{k,k+1}$ in
$H$ with $Y^\pm_i$ and $L_j$ follows from Proposition
\eref{Propdiag} and the symmetry \eref{symmetryG} of the Perk-Schultz quantum chain. 

To prove the explicit expressions \eref{repLj}, \eref{repYi} for the representations 
we focus on $Y^+_1(r)$. Using the
fundamental representation and the factorization property \eref{tensorfactor}
of the inner product of tensor vectors 
one finds $\bra{\eta'} a_r^+ \ket{\eta} = \delta_{\eta',\eta^{r-}}$.
The terms in the exponential follow trivially from Lemma \eref{projlem}
and the definitions \eref{partnum2}, \eref{NyMx}. Following similar arguments
for the other generators yields the
matrix elements
of $Y^\pm_i$ and $L_j$ \eref{repLj}, \eref{repYi} as stated in \eref{Theosymmetry}. \qed
\end{proof}

\subsection{Proof of Theorem \eref{Theoinvmeas}}

\begin{proof}
Since $G$ is symmetric and $R$ is diagonal, Proposition \eref{Propdiag} implies
\bel{HT}
H^T = (RR^T)^{-1} H RR^T = R^{-2} H R^2.
\ee
With \eref{revtrans2} we thus have reversibility 
with a reversible measure $\hat{\pi} = R^2$ in the matrix form
\eref{invmeasmatrix}.
By the projector lemma \eref{projlem} $\hat{\pi}$ yields the reversible measure
\eref{theo1}
of Theorem \eref{Theoinvmeas}. \qed
\end{proof}

\begin{acknowledgement}
This work was supported by DFG and by CNPq through the grant 307347/2013-3.
GMS thanks the University of S\~ao Paulo,
where part of this work was done, for kind hospitality.
\end{acknowledgement}

\section*{Appendix}
\addcontentsline{toc}{section}{Appendix}

We display some explicit results for unnormalized stationary distributions
for small lattices $L=2,3,4$ and also $L$ arbitrary with small particle numbers $N+M = 1,2,3,4$.
$$
\ba{lll}
\hline
N=0 & M=0 & \ket{00}  \hspace*{3cm} L=2 \hspace*{5cm} \\
\hline \\[-4mm]
N=1 & M=0 & q^{-1} \ket{A0} + q \ket{0A}  \\
N=0 & M=1 & q \ket{B0} + q^{-1} \ket{0B} \\[1mm]
\hline \\[-4mm]
N=2 & M=0 &  \ket{AA}  \\
N=1 & M=1 & q^{-1} \ket{AB} + q \ket{BA}  \\
N=0 & M=2 &  \ket{BB}   \\
\hline \\[-4mm]
& &
\ea
$$
$$
\ba{lll}
\hline
N=0 & M=0 & \ket{000} \hspace*{3cm} L=3 \\
\hline \\[-4mm]
N=1 & M=0 & q^{-2} \ket{A00} + \ket{0A0} + q^{2} \ket{00A} \\
N=0 & M=1 & q^{-2} \ket{00B} + \ket{0B0} + q^{2} \ket{B00} \\
\hline \\[-4mm]
N=2 & M=0 &  q^{-2} \ket{AA0} + \ket{A0A} + q^{2} \ket{0AA}  \\
N=1 & M=1 & q^{-3} \ket{A0B} + q^{-1} (\ket{AB0}+\ket{0AB}) 
+ q (\ket{BA0}+\ket{0BA}) + q^{3} \ket{B0A}  \\
N=0 & M=2 &  q^{-2} \ket{0BB} + \ket{B0B} + q^{2} \ket{BB0}  \\
\hline \\[-4mm]
N=3 & M=0 & \ket{AAA}  \\
N=2 & M=1 & q^{-2} \ket{AAB} + \ket{ABA} + q^{2} \ket{BAA}  \\
N=1 & M=2 &  q^{-2} \ket{ABB} + \ket{BAB} + q^{2} \ket{BBA}  \\
N=0 & M=3 & \ket{BBB} \\
\hline \\[-4mm]
& &
\ea
$$
$$
\ba{lll}
\hline
N=0 & M=0 & \ket{0000}  \hspace*{3cm} L=4 \\
\hline \\[-4mm]
N=1 & M=0 & q^{-3} \ket{A000} + q^{-1} \ket{0A00} + q \ket{00A0} + q^{3} \ket{000A}\\
N=0 & M=1 & q^{-3} \ket{000B} + q^{-1} \ket{00B0} + q \ket{0B00} + q^{3} \ket{B000}\\
\hline \\[-4mm]
N=2 & M=0 &  q^{-4} \ket{AA00} + q^{-2} \ket{A0A0} + \ket{A00A} + \ket{0AA0}  + q^{2} \ket{0A0A}+ q^{4} \ket{00AA}\\
N=1 & M=1 & q^{-5} \ket{A00B} + q^{-3} (\ket{A0B0} + \ket{0A0B}) + q^{-1} (\ket{AB00} + \ket{0AB0} + \ket{00AB}) \\
& & + q (\ket{BA00} + \ket{0BA0} + \ket{00BA})+ q^{3} (\ket{B0A0} + \ket{0B0A}) + q^{5} \ket{B00A}\\
N=0 & M=2 &  q^{-4} \ket{00BB} + q^{-2} \ket{0B0B} + \ket{B00B} + \ket{0BB0}  + q^{2} \ket{B0B0}+ q^{4} \ket{BB00}\\
\hline \\[-4mm]
N=3 & M=0 & q^{-3} \ket{AAA0} + q^{-1} \ket{AA0A} + q \ket{A0AA} + q^{3} \ket{0AAA}\\
N=2 & M=1 & q^{-5} \ket{AA0B} + q^{-3} (\ket{AAB0} + \ket{A0AB}) + q^{-1} (\ket{0AAB} + \ket{ABA0} + \ket{A0BA})\\
& & + q (\ket{AB0A} + \ket{0ABA} + \ket{BAA0}) + q^{3} (\ket{BA0A} + \ket{0BAA}) + q^{5} \ket{B0AA} \\
N=1 & M=2 & q^{-5} \ket{A0BB} + q^{-3} (\ket{AB0B} + \ket{0ABB}) + q^{-1} (\ket{BA0B} + \ket{0BAB} + \ket{ABB0})\\
& & + q (\ket{0BBA} + \ket{BAB0} + \ket{B0AB}) + q^{3} (\ket{BBA0} + \ket{B0BA}) + q^{5} \ket{BB0A} \\
N=0 & M=3 & q^{-3} \ket{0BBB} + q^{-1} \ket{B0BB} + q \ket{BB0B} + q^{3} \ket{BBB0}\\
\hline \\[-4mm]
N=4 & M=0 & \ket{AAAA}  \\
N=3 & M=1 & q^{-3} \ket{AAAB} + q^{-1} \ket{AABA} + q \ket{ABAA} + q^{3} \ket{BAAA}\\
N=2 & M=2 &  q^{-4} \ket{AABB} + q^{-2} \ket{ABAB} + \ket{ABBA} + \ket{BAAB}  + q^{2} \ket{BABA}+ q^{4} \ket{BBAA}\\
N=1 & M=3 & q^{-3} \ket{ABBB} + q^{-1} \ket{BABB} + q \ket{BBAB} + q^{3} \ket{BBBA}\\
N=0 & M=4 & \ket{BBBB}  \\
\hline \\[-4mm]
& &
\ea
$$

\noindent \underline{$N+M=1$:} 
\bea
\pi^\ast(\{x\},\emptyset) & \propto & q^{2x-1}  \nonumber \\
\pi^\ast(\emptyset,\{y\}) & \propto & q^{-2y+1}  \nonumber 
\eea

\noindent \underline{$N+M=2$:} 
\bea
\pi^\ast(\{x_1,x_2\},\emptyset) & \propto & q^{2 x_1 + 2 x_2 - 2}\nonumber \\
\pi^\ast(\{x_1\},\{y_1\}) & \propto & 
\left\{ \ba{ll} q^{2 x_1 - 2 y_1 - 1} & \quad y_1<x_1 \\  q^{2 x_1 - 2 y_1 + 1} & \quad y_1>x_1 \ea \right. \nonumber \\
\pi^\ast(\emptyset,\{y_1,y_2\}) & \propto & q^{-2 y_1 - 2 y_2 + 2} \nonumber 
\eea

\noindent \underline{$N+M=3$:}
\bea
\pi^\ast(\{x_1,x_2,x_3\},\emptyset) & \propto & q^{2 x_1 + 2 x_2 + 2 x_3 - 3}\nonumber \\
\pi^\ast(\{x_1,x_2\},\{y_1\}) & \propto & \left\{ \ba{ll} q^{2 x_1 + 2 x_2 - 2 y_1 - 3} & \quad y_1<x_1,x_2 \\  
q^{2 x_1 + 2 x_2 - 2 y_1 - 1} & \quad x_1< y_1<x_2 \\
q^{2 x_1 + 2 x_2 - 2 y_1 + 1} & \quad x_1,x_2< y_1
\ea \right. \nonumber \\
\pi^\ast(\{x_1\},\{y_1,y_2\}) & \propto & \left\{ \ba{ll} q^{2 x_1 - 2 y_1 - 2 y_2 - 1} & \quad y_1,y_2<x_1 \\  
q^{2 x_1 - 2 y_1 - 2 y_2 + 1} & \quad y_1< x_1<y_2 \\
q^{2 x_1 - 2 y_1 - 2 y_2 + 3} & \quad x_1< y_1,y_2
\ea \right. \nonumber \\
\pi^\ast(\emptyset,\{y_1,y_2,y_3\}) & \propto & q^{- 2 y_1 - 2 y_2 - 2 y_3 + 3}\nonumber
\eea

\noindent \underline{$N+M=4$:}

\bea
\pi^\ast(\{x_1,x_2,x_3,x_4\},\emptyset) & \propto & q^{2 x_1 + 2 x_2 + 2 x_3 + 2 x_4 - 4}\nonumber \\
\pi^\ast(\{x_1,x_2,x_3\},\{y_1\}) & \propto & \left\{ \ba{ll} 
q^{2 x_1 + 2 x_2 + 2 x_3 - 2 y_1 - 5} & \quad y_1<x_1,x_2,x_3 \\  
q^{2 x_1 + 2 x_2 + 2 x_3 - 2 y_1 - 3} & \quad x_1< y_1<x_2,x_3 \\
q^{2 x_1 + 2 x_2 + 2 x_3 - 2 y_1 - 1} & \quad x_1,x_2< y_1<x_3 \\
q^{2 x_1 + 2 x_2 + 2 x_3 - 2 y_1 + 1} & \quad x_1,x_2,x_3< y_1
\ea \right. \nonumber \\
\pi^\ast(\{x_1,x_2\},\{y_1,y_2\}) & \propto & \left\{ \ba{ll} 
q^{2 x_1 + 2x_2 - 2 y_1 - 2 y_2 - 4} & \quad y_1,y_2<x_1,x_2 \\  
q^{2 x_1 + 2x_2 - 2 y_1 - 2 y_2 - 2} & \quad y_1<x_1<y_2<x_2 \\  
q^{2 x_1 + 2x_2 - 2 y_1 - 2 y_2}      & \quad y_1<x_1,x_2<y_2 \\  
q^{2 x_1 + 2x_2 - 2 y_1 - 2 y_2}      & \quad x_1<y_1,y_2<x_2 \\  
q^{2 x_1 + 2x_2 - 2 y_1 - 2 y_2 + 2} & \quad x_1<y_1<x_2<y_2 \\  
q^{2 x_1 + 2x_2 - 2 y_1 - 2 y_2 + 4} & \quad x_1,x_2<y_1,y_2 
\ea \right. \nonumber \\
\pi^\ast(\{x_1\},\{y_1,y_2,y_3\}) & \propto & \left\{ \ba{ll} 
q^{2 x_1 - 2 y_1 - 2 y_2 - 2 y_3 - 1} & \quad y_1,y_2,y_3 < x_1\\  
q^{2 x_1 - 2 y_1 - 2 y_2 - 2 y_3 + 1} & \quad y_1,y_2< x_1<x_3  \\
q^{2 x_1 - 2 y_1 - 2 y_2 - 2 y_3 + 3} & \quad y_1< x_1<y_2,y_3 \\
q^{2 x_1 - 2 y_1 - 2 y_2 - 2 y_3 + 5} &  \quad x_1 < y_1,y_2,y_3
\ea \right. \nonumber \\
\pi^\ast(\emptyset,\{y_1,y_2,y_3,y_4\}) & \propto & q^{- 2 y_1 - 2 y_2 - 2 y_3 - 2 y_4 + 4}\nonumber
\eea

\end{document}